\documentclass[runningheads]{llncs}
\usepackage{fullpage}

\usepackage[comma,authoryear,round,sectionbib]{natbib}

\usepackage{graphicx}
\usepackage{wrapfig}
\usepackage{booktabs,complexity,tabularx}
\usepackage{amsmath,amsfonts,amssymb}
\usepackage{stmaryrd}
\usepackage[boxed]{algorithm}
\usepackage[noend]{algorithmic}

\usepackage{turnstile}
\usepackage{mathrsfs}
\usepackage{enumitem}
\usepackage{amsthm}
\usepackage[super]{nth}
\usepackage{listings}
\usepackage{mathtools}
\usepackage{csquotes}
\usepackage{multicol,multirow}
\usepackage{fancyvrb}
\usepackage{makecell}
\usepackage{mdframed}
\usepackage[toc,page]{appendix}
\usepackage{mathdots}
\usepackage{empheq}
\usepackage{bm}
\usepackage{pgfplots}
\usepackage{enumitem}
\usepackage{float}
\usepackage{nicefrac} 
\usepackage{array}
\usepackage{blkarray}

\usepackage{tikz}
\usetikzlibrary{automata, positioning, calc, shapes, arrows, fit}

\newenvironment{claimproof}{\par\noindent\underline{Proof:}}{\leavevmode\unskip\penalty9999 \hbox{}\nobreak\hfill\quad\hbox{$\blacksquare$}}

\usepackage[unicode=true, bookmarks=false, breaklinks=true, pdfborder={0 0 1}, colorlinks=true, citecolor=blue]{hyperref}

\usepackage{pifont}
\definecolor{MyGreen}{rgb}{0, 0.7, 0}
\definecolor{MyRed}{rgb}{0.8, 0, 0}

\usepackage{caption}
\newcommand{\maj}{\textsc{maj-hai}}
\newcommand{\cmaj}{\textsc{maj-hac}}
\newcommand{\tmaj}{\textsc{maj-hati}}
\newcommand{\ctmaj}{\textsc{maj-hatc}}
\newcommand{\khati}{$c$-\textsc{hati}}
\newcommand{\khai}{$c$-\textsc{hai}}
\newcommand{\khat}{$c$-\textsc{hatc}}
\newcommand{\maxmajsmi}{\textsc{max-maj-smi}}
\newcommand{\ksm}{$c$-\textsc{smc}}
\newcommand{\ksr}{$c$-\textsc{src}}
\newcommand{\ksri}{$c$-\textsc{sri}}
\newcommand{\ksmi}{$c$-\textsc{smi}}
\newcommand{\kha}{$c$-\textsc{hac}}

\usepackage[textwidth=20mm,obeyFinal,colorinlistoftodos,textsize=small]{todonotes}

\usepackage{boxedminipage}
\usepackage{xspace}

\newcommand{\pbDef}[3]{%
\noindent
\begin{center}
\begin{boxedminipage}{0.98\columnwidth}
#1\\[5pt]
\begin{tabular}{p{0.14\columnwidth}p{0.8\columnwidth}}
Input: & #2\\
Question: & #3
\end{tabular}
\end{boxedminipage}
\end{center}
}

\newtheorem{observation}[theorem]{Observation}

\newtheorem{mycorollary}[theorem]{Corollary}

\newcommand{\midd}{\mathbin{:}}

\usepackage{mathtools}

\title{Computational complexity of $k$-stable matchings}

 \author{Haris Aziz\inst{1} \and Gergely Cs\'{a}ji\inst{2,3} \and \'{A}gnes Cseh\inst{3,4}}
 \authorrunning{H. Aziz et al.}
 \institute{UNSW Sydney, Australia 
 \and
 E\"{o}tv\"{o}s Lor\'and University, Budapest, Hungary
 \and 
 Institute of Economics, 
  Centre for Economic and Regional Studies, Budapest, Hungary
 \and
 University of Bayreuth, Germany
}


\begin{document}
	\maketitle
\begin{abstract}
We study deviations by a group of agents in the three main types of matching markets: the house allocation, the marriage, and the roommates models. For a given instance, we call a matching \emph{$k$-stable} if no other matching exists that is more beneficial to at least $k$ out of the $n$ agents. The concept generalizes the recently studied majority stability \citep{Tha21}. We prove that whereas the verification of $k$-stability for a given matching is polynomial-time solvable in all three models, the complexity of deciding whether a $k$-stable matching exists depends on $\frac{k}{n}$ and is characteristic to each model.
\end{abstract}

\section{Introduction}

In matchings under preferences, agents seek to be matched among themselves or to objects. Each agent has a preference list on their possible partners. When an agent is asked to vote between two offered matchings, they vote for the one that allocates the more desirable partner to them. The goal of the mechanism designer is to compute a matching that guarantees some type of optimality. A rich literature has emerged from various combinations of input type and optimality condition. In our paper, we study three classic input types together with a new, flexible optimality condition that incorporates already defined notions as well.

\paragraph{Input types.} Our three input types differ in the structure of the underlying graph and the existence of objects as follows.
\begin{itemize}
\item {\em House allocation model.} One side of a two-sided matching instance consists of agents who have strictly ordered, but possibly incomplete preferences and cast votes, while the other side is formed by objects with no preferences or votes.
\item {\em Marriage model.} Vertices on both sides of a two-sided matching instance are agents, who all have strictly ordered, but possibly incomplete preferences and cast votes. \item {\em Roommates model.} The matching instance is not necessarily two-sided, all vertices are agents, who have strictly ordered, but possibly incomplete preferences and cast votes.
\end{itemize}

\paragraph{Optimality condition.} For a given $k$, we say that a matching $M$ is \emph{$k$-stable} if there is no other matching $M'$ that at least $k$ agents prefer to~$M$. Notice that this notion is highly restrictive, as the number of agents who prefer $M$ to $M'$ is not taken into account. Some special cases of $k$ express very intuitive notions. The well-known notion of weak Pareto optimality is equivalent to $n$-stability; majority stability \citep{Tha21} is equivalent to $\frac{n+1}{2}$-stability, and finally, 1-stability asks whether there is a matching that assigns each agent their most preferred partner.

\paragraph{Structure of the paper and techniques.} We summarize relevant known results in Section~\ref{se:related} and lay the formal foundations of our investigation in Section~\ref{se:preliminaries}. We then turn to our complexity results for the house allocation model in Section~\ref{se:ha} and provide analogous proofs for the marriage and roommates models in Section~\ref{se:mr}. We conclude in Section~\ref{se:con}. Our proofs rely on tools from matching theory such as the famous Gallai-Edmonds decomposition, stable partitions, or scaling an instance with carefully designed gadgets.
 
\section{Related work}
\label{se:related}

Matchings under preferences have been actively researched by both Computer Scientists and Economists \citep{Rot82,Man13}. In this section, we highlight known results on the most closely related optimality concepts from the field.

\subsection{Pareto optimal matchings}

Pareto optimality is a desirable condition, most typically studied in the house allocation model. It is often combined with other criteria, such as lower and upper quotas. A matching $M$ is \emph{Pareto optimal} if there is no matching $M'$, in which no agent is matched to a object they consider worse, while at least one agent is matched to a object they consider better than their object in~$M$. A much less restrictive requirement implies \emph{weak Pareto optimality}: $M$ is weakly Pareto optimal if no matching $M'$ exists that is preferred by all agents. This notion is equivalent to $n$-stability.

Weak Pareto optimality is mainly used in continuous and multi-objective optimization \citep{EN02} and in economic theory~\citep{War83,FGJ06}. Pareto optimality is one of the most studied concepts in coalition formation and hedonic games \citep{ABH13,B20,EFF20,BFO22}, and has also been defined in the context of various matching markets \citep{CEFM+14,CEFM+16,ACGS18,BG20}. 
As shown by \citet{ACM+04}, in the house allocation model, a maximum size Pareto optimal matching can be found in polynomial time.

\subsection{Stable matchings}

Possibly the most studied optimality notion for the marriage and roommates models is stability. A matching is \emph{stable} if it is not \emph{blocked} by any edge, that is, no pair of agents exists who are mutually inclined to abandon their partners for each other. There is a striking difference between the definitions of stability and $k$-stability: while the existence of a blocking edge is an inherently local property, $k$-stability is defined in relation with other matchings. 

The existence of stable matchings was shown in the seminal paper of \citet{GS62} for the marriage model. Later, \citet{Irv85} gave a polynomial-time algorithm to decide whether a given roommates instance admits a stable matching. \cite{tan1991necessary} improved Irving's algorithm by providing an algorithm that always finds a so-called stable partition, which coincides with a stable matching if any exists. Stability was later extended to various other input settings in order to suit the growing number of applications such as  employer matching markets \citep{RS90}, university admission decisions \citep{BS99,BDK10}, campus housing matchings \citep{CS02,PPR08}, and bandwidth matching \citep{GLM+07}.

\subsection{Popular matchings}

Popular matchings translate the simple majority voting rule into the world of matchings under preferences. Given two matchings~$M$ and $M'$, matching $M$ is more popular than $M'$ if the number of vertices preferring $M$ to $M'$ is larger than the number of vertices preferring $M'$ to~$M$. A matching~$M$ is \emph{popular} in an instance if there is no matching~$M'$ that is more popular than~$M$. The main difference between $k$-stability and popularity is that in the earlier, only agents are counted who benefit from switching to an alternative matching, while in popularity, agents can vote both for and against the alternative matching.

The concept of popularity was first introduced by \citet{Gar75} for the marriage model, and then studied by \citet{AIK+07} in the house allocation model. Polynomial-time algorithms to find a popular matching were given in both models. In the marriage model, it was already noticed by \citeauthor{Gar75} that all stable matchings are popular, which implies that in this model, popular matchings always exist. In fact stable matchings are the smallest size popular matchings, as shown by \citet{BIM10}, while maximum size popular matchings can be found in polynomial time as well \citep{HK13,Kav14}. Only recently \citet{FKP+19} and \citet{GMS+21} resolved the long-standing \citep{BIM10,HK13a,Man13,Cse17,HK21} open problem on the complexity of deciding whether a popular matching exists in a popular roommates instance and showed that the problem is $\NP$-complete. This hardness extends to graphs with complete preference lists \citep{CK21}.

Besides the three matching models, popularity has also been defined for spanning trees \citep{Dar13}, permutations \citep{VSW14,KCM21}, the ordinal group activity selection problem \citep{Dar18}, and very recently, for branchings \citep{KKM+22a}. Matchings nevertheless constitute the most actively researched area of the majority voting rule outside of the usual voting scenarios.

\subsection{Relaxing popularity}
The two most commonly used notions for near-popularity are called \emph{minimum unpopularity factor} \citep{McC08,KMN11,HK13,Kav14,BHH+15,KKM+22a,RI21} and \emph{minimum unpopularity margin} \citep{McC08,HKMN11,KKM+22a,KKM+22b}. Both notions express that a near-popular matching is never beaten by too many votes in a pairwise comparison with another matching. We say that matching $M'$ \emph{dominates} matching $M$ by a margin of $u-v$, where $u$ is the number of agents who prefer $M'$ to $M$, while $v$ is the number of agents who prefer $M$ to~$M'$. The \emph{unpopularity margin} of $M$ is the maximum margin by which it is dominated by any other matching. As opposed to $k$-stability, the  unpopularity margin takes the number of both the satisfied and dissatisfied agents into account when comparing two matchings.

Checking whether a matching $M'$ exists that dominates a given matching $M$ by a margin of $k$ can be done in polynomial time by the standard popularity verification algorithms in all models \citep{AIK+07,BIM10,RI21}. Finding a least-unpopularity-margin matching in the house allocation model is NP-hard \citep{McC08}, which implies that for a given (general) $k$, deciding whether a matching with unpopularity margin $k$ exists is also NP-complete. A matching of unpopularity factor~0, which is a popular matching, always exists in the marriage model, whereas deciding whether such a matching exists in the roommates model is NP-complete \citep{FKP+19,GMS+21}. 

The unpopularity margin of a matching expresses the degree of undefeatability of a matching admittedly better than our $k$-stability. We see a different potential in $k$-stability and majority stability. The fact that, compared to $M$, there is no alternative matching in which at least $k$ agents improve simultaneously, is a strong reason for choosing $M$---especially if $k = \frac{n+1}{2}$. 
The decision maker might care about minimizing the number of agents who would mutually improve by switching to an alternative matching. If there is a matching, where a significant number of agents can improve simultaneously, then they may protest together against the central agency to change the outcome---even though it would make some other agents worse off---out of ignorance or lack of information about the preferences of others. The unpopularity margin and factor give no information on this aspect, as they only measure the relative number of improving and disimproving agents.

\subsection{Majority stability}
The study of majority stable matchings was initiated very recently by \citet{Tha21}. The three well-known voting rules plurality, majority, and unanimity translate into popularity, majority stability, and Pareto optimality in the matching world. A matching $M$ is called \emph{majority stable} if no matching $M'$ exists that is preferred by more than half of all voters to~$M$. The concept is equivalent to $\frac{n+1}{2}$-stability in our terminology.

\citeauthor{Tha21} observed that majority stability, in sharp contrast to popularity, is strikingly robust to correlated preferences. Based on this, he argued that in application areas where preferences are interdependent, majority stability is a more desirable solution concept than popularity. He provided examples and simulations to illustrate that, unlike majority stable matchings, the existence of a popular matching is sensitive to even small levels of correlations across individual's preferences. Via a linear programming approach he also showed that the verification of majority stability is polynomial-time solvable in the house allocation model.

\section{Preliminaries}
\label{se:preliminaries}

In this section, we describe our input settings, formally define our optimality concepts, and give a structured overview of all investigated problems. 

\subsection{Input}
In the simplest of our three models, the house allocation model, we consider a set of agents $N=\{1,\ldots,n\}$ and a set of objects~$O$. Each agent $i\in N$ has strict preferences $\succ_i$ over a subset of $O$, called the set of \emph{acceptable} objects to $i$, while objects do not have preferences. The notation $o_1 \succ_i o_2$ means that agent $i$ prefers object $o_1$ to object~$o_2$. Being unmatched is considered worse by agents than being matched to any acceptable object or agent. To get a more complete picture, we also explore cases, where ties are allowed in the preference lists. A \emph{matching} assigns each object to at most one agent and gives at most one acceptable object to each agent. 

In the marriage model, no objects are present. Instead, the agent set $N = U \cup W$ is partitioned into two disjoint sets, and each agent seeks to be matched to an acceptable agent from the other set. In the roommates model, an agent from the agent set $N$ can be matched to any acceptable agent in the same set.

For clearer phrasing, we often work in a purely graph theoretical context. The \emph{acceptability graph} of an instance consists of the agents and objects as vertices and the acceptability relations as edges between them. This graph is bipartite in the house allocation and marriage models.

For a matching $M$, we denote by $M(i)$ the object or agent assigned to agent $i\in N$. Each agent's preferences over objects or agents can be extended naturally to corresponding preferences over matchings. According to these extended preferences, an agent is indifferent among all matchings in which they are assigned to the same object or agent. Furthermore, agent $i$ prefers matching $M'$ to matching $M$ if $M'(i) \succ_i M(i)$. 



\subsection{Optimality}
Next, we define some standard optimality concepts from the literature. A matching $M$ is 

\begin{itemize}
	\item \emph{weakly Pareto optimal} if there exists no other matching $M'$ such that $M'(i) \succ_i M(i)$ for all $i\in N$;
	\item 
 \emph{majority stable} if there exists no other matching $M'$ such that
		$|i\in N\midd M'(i)\succ_iM(i)|\geq \frac{n+1}{2}$;
	\item 
 \emph{popular} if there exists no other matching $M'$ such that
	$|i\in N\midd M'(i)\succ_iM(i)|>|i\in N\midd M(i)\succ_iM'(i)|$.
\end{itemize} 

In words, weak Pareto optimality means that, compared to $M$, no matching is better for all agents, majority stability means that no matching is better for a majority of all agents, while popularity means that no matching is better for a majority of the agents who are not indifferent between the two matchings. It is easy to see that popularity implies majority stability, which in turn implies weak Pareto optimality. 

We refine this scale of optimality notions by adding $k$-stability to it. A matching $M$ is 
\begin{itemize}
	\item \emph{$k$-stable} if there exists no other matching $M'$ such that
		$|i\in N\midd M'(i)\succ_iM(i)|\geq k.$
\end{itemize}
In words, $k$-stability means that no matching $M'$ is better for at least $k$ agents than $M$---regardless of how many agents prefer $M$ to~$M'$. Weak Pareto optimality is equivalent to $n$-stability, while majority stability is equivalent to $\frac{n+1}{2}$-stability. It follows from the definition that $k$-stability implies $\left(k+1\right)$-stability.

We demonstrate $k$-stability on an example instance, which we will also use in our proofs later.

\begin{example}[An $(n-1)$-stable matching may not exist]
\label{ex:cond}
Consider an instance in which $N=\{1,2, \ldots, n\}$ and $O=\{o_1,o_2, \ldots, o_n\}$. Each agent has identical preferences of the form $o_1\succ o_2 \succ \ldots \succ o_n$, analogously to the preferences in the famous example of \citet{Con85}. For an arbitrary matching $M$, each agent $i$ of the at least $n-1$ agents, for whom $M(i) \neq o_1$ holds, could  improve by switching to the matching that gives them the object directly above $M(i)$ in the preference list (or any object if $i$ was unmatched in~$M$). Therefore, no matching is majority stable or $(n-1)$-stable.
\end{example}

\subsection{Our problems and contribution}

Now we define our central decision problems formally. We are particularly interested in the computational complexity and existence of $k$-stable matchings depending on the value $c=\frac{k}{n}$, so in our problems we investigate $cn$-stability for all constants $c\in (0,1)$. One problem setting is below.

\pbDef{\khai}{
Agent set $N$, object set $O$, a strict ranking $\succ_i$ over the acceptable objects for each $i \in N$ and a constant $c\in (0,1)$.
}{
Does a $cn$-stable matching exist?}

In most hardness results, we will choose the constant $c$ not to be part of the input, but to be a universal constant instead.
Our further problem names also follow the conventions \citep{Man13}. For majority stability instead of $cn$-stability, we add the prefix \textsc{maj}. We substitute \textsc{ha} by \textsc{sm} for the marriage model, and by \textsc{sr} for the roommates model. If the preference lists are complete, that is, if all agents find all objects acceptable, then we replace the \textsc{i} standing for incomplete by a \textsc{c} standing for complete. If ties are allowed in the preference lists, we add a \textsc{t}. Table~\ref{ta:names} depicts a concise overview of the problem names.

\begin{table}[htb]
      \centering
      \setlength{\tabcolsep}{10pt}
\begin{tabular}{c|c|c|c}
           Optimality criterion & Model  & Presence of ties & Completeness of preferences \\
           \hline\
           $c$ or \textsc{maj} & \textsc{ha} or \textsc{sm} or \textsc{sr} & \textsc{t} or $\emptyset$ & \textsc{c} or \textsc{i} 
        \end{tabular}
\caption{Each problem name consists of four components, as shown in the columns of the table.}
\label{ta:names}
\end{table}

For each of the two optimality criteria, there are $3\cdot 2 \cdot 2 = 12$ problem variants. 
Our goal was to solve all 12 variants for all $0\le c\le 1$ values, and also majority stability (which corresponds to $cn$-stability for $c=\frac{1}{2}+\varepsilon$, if $\varepsilon <\frac{1}{n}$). Our results are summarized in Table~\ref{ta:results}. We remark that our positive results are also existencial results, that is we show that a majority stable matching always exists in the \textsc{sm} model and a $\left(\frac{5}{6}n+\varepsilon\right)$-stable matching always exists in the \textsc{sr} model for any $\varepsilon >0$.
We only leave open the complexity of the four variants of $c$-\textsc{sr} for $\frac{2}{3}< c\le \frac{5}{6}$. 

\setlength\extrarowheight{3pt}
\begin{table}[ht]
\centering
    \resizebox{\textwidth}{!}{
\begin{tabular}{|c||c|c|c|c|c|}
   \noalign{\hrule}
Problem &\textsc{ha} & \multicolumn{2}{|c}{\textsc{sm}} & \multicolumn{2}{|c|}{\textsc{sr}} \\
\cline{2-6}
variant & NP-complete & NP-complete for $c \le \frac{1}{2}$ & P for $c > \frac{1}{2}$ & NP-complete for $c \le  \frac{2}{3}$ & P for $c > \frac{5}{6}$ \\
\noalign{\hrule}
\textsc{c}, \textsc{i}, \textsc{tc}, \textsc{ti}  & Theorems~\ref{th:khai_allk}, \ref{th:kha}        & Theorem~\ref{th:npsmsr}   & Theorem~\ref{th:smpol}     &  Theorem~\ref{th:npsmsr}  & Theorem~\ref{th:ksrpol} \\
   \noalign{\hrule}
\end{tabular}
}
\caption{Our results on the complexity of deciding whether a $k$-stable matching exists.}\label{ta:results}
\end{table}

In order to draw a more accurate picture in the presence of ties, we also investigate two standard input restrictions \citep{BM04,Pet16,CHK17}, see Table~\ref{ta:results2}.
\begin{itemize}
    \item  \textsc{dc}: dichotomous and complete preferences, which means that agents classify all objects or other agents as  either ``good'' or ``bad'', and can be matched to either one of these.
    \item \textsc{sti}: Possibly incomplete preferences consisting of a single tie, which again means that agents classify all objects or other agents as  either ``good'' or ``bad'', and consider a bad match to be unacceptable.
\end{itemize}
Note that these two cases can be different, because if the preferences are \textsc{dc}, then each agent can be in one of three situations: matched to a good partner, matched to a bad partner or remain unmatched. In \textsc{sti}, each agent is either matched or unmatched.

\begin{table}[htb]
\centering
    \resizebox{\textwidth}{!}{
\begin{tabular}{|c||c|c|c|c|c|c|}
   \noalign{\hrule}
Problem &\multicolumn{3}{c}{\textsc{ha}} & \multicolumn{3}{|c|}{\textsc{sm} and \textsc{sr}}  \\
\cline{2-7}
variant &  $c\le \frac{1}{2}$ &  $\frac{1}{2}<c\le \frac{2}{3}$  &  $c>\frac{2}{3}$  &    $c\le \frac{1}{3}$&  $\frac{1}{3}<c \le \frac{1}{2}$ & $c > \frac{1}{2}$   \\
\noalign{\hrule}
 \textsc{dc}  &    NP-c: Theorem~\ref{th:khat}   & NP-c: Theorem~\ref{th:khat}  & \multirow{2}{*}{P: Theorem~\ref{thm:alltie+dichHA}}  &  NP-c: Theorem~\ref{thm:dichcompsrneg}  & NP-c: Theorem~\ref{thm:dichcompsrneg} & P: Theorem~\ref{thm:dichcompsrpos}   \\
\textsc{sti}  &   NP-c: Theorem~\ref{th:ksmallties}   & P: Theorem~\ref{thm:alltie+dichHA}  & &   NP-c: Theorem~\ref{thm:alltiesrneg}  & P: Theorem~\ref{thm:alltiesrpos}  &   P: Theorem~\ref{thm:alltiesrpos} \\
\hline
\end{tabular}
}
\caption{Results for the two restricted settings in the presence of ties. NP-c abbreviates NP-complete.}
\label{ta:results2}
\end{table}

\section{The house allocation model}
\label{se:ha}
In this section, we examine the computational aspects of $k$-stability and majority stability in the house allocation model. We first present positive results  on verification in Section~\ref{se:haver} and then turn to hardness proofs on existence and some solvable restricted cases in Section~\ref{se:haex}.

\subsection{Verification}
\label{se:haver}
\citet{Tha21} constructed an integer linear program to check whether a given matching is majority stable. He showed that the underlying matrix of the integer linear constraints is unimodular and hence the problem can be solved in polynomial time. Here we provide a simple characterization of majority stable matchings, which also delivers a fast and simple algorithm for testing majority stability.

In this subsection, we investigate $k$-stability for an arbitrary, but fixed $k \in \mathbb{N}$, which does not depend on~$n$. Our first observation characterizes $1$-stable matchings.
							
\begin{observation}
    A matching is $1$-stable if and only if each agent gets their most preferred object.
\end{observation}						

To generalize this straightforward observation to $k$-stability and majority stability, we introduce the natural concept of an improvement graph. For a given matching $M$, let $G_M=(N\cup O,E)$ be the corresponding \textit{improvement graph}, where $(i,o)\in E$ if and only if $o\succ_i M(i)$. In words, the improvement graph consists of edges that agents prefer to their current matching edge. We also say that agent $i$ \emph{envies} object $o$ if $(i,o)\in E$.

\begin{observation}\label{lemma:ks-charac}
	Matching $M$ is $k$-stable if and only if $G_M$ does not admit a matching of size at least~$k$. In particular, $M$ is majority stable if and only if $G_M$ does not admit a matching of size at least~$\frac{n+1}{2}$.
\end{observation}
\begin{proof}
 It follows from the definition of $G_M$ that $G_M$ admits a matching $M'$ of size at least~$k$ if and only if in $M'$, at least $k$ agents get a better object than in~$M$. The non-existence of such a matching $M'$ defines $k$-stability for~$M$. 
\end{proof}
		
    

Observation~\ref{lemma:ks-charac} delivers a polynomial verification method for checking $k$-stability and majority stability. Constructing $G_M$ to a given matching $M$ takes at most $O(m)$ time, where $m$ is the number of acceptable agent-object pairs in total. Finding a maximum size matching in $G_M$ takes $O(\sqrt{n}m)$ time \citep{HK73}.

\begin{mycorollary}
\label{cor:k-verif}
	For any $k \in \mathbb{N}$, it can be checked in $O(\sqrt{n}m)$ time whether a given matching is $k$-stable. In particular, verifying majority stability can be done in $O(\sqrt{n}m)$ time.
\end{mycorollary}

\subsection{Existence}
\label{se:haex}

By Corollary~\ref{cor:k-verif}, all decision problems on the existence of a $k$-stable matching in the house allocation model are in NP. Our hardness proofs rely on reductions from the problem named exact cover by 3-sets (\textsc{x3c}), which was shown to be NP-complete by \citet{GJ79}. First we present our results for the problem variants with possibly incomplete preference lists, then extend these to complete lists, and finally we discuss the case of ties in the preferences.

\pbDef{\textsc{x3c}}{
A set $\mathcal{X}=\{ 1,\dots,3\hat{n}\}$ and a family of 3-sets $\mathcal{S}\subset \mathcal{P}(\mathcal{X})$ of cardinality $3\hat{n}$ such that each element in $\mathcal{X}$ is contained in exactly three sets.}{
Are there $\hat{n}$ 3-sets that form an exact 3-cover of $\mathcal{X}$, that is, each element in $\mathcal{X}$ appears in exactly one of the $\hat{n}$ 3-sets?}

\subsubsection{Incomplete preferences}
\begin{theorem}
\label{th:khai}
\khai\ is NP-complete even if each agent finds at most two objects acceptable.
\end{theorem}
\begin{proof}
Let $I$ be an instance of \textsc{x3c}, where $\mathcal{S}=\{ S_1,\dots,S_{3\hat{n}}\}$ is the family of 3-sets and $S_j=\{ j_1,j_2,j_3\}$. We build an instance $I'$ of \khai\ as follows. For each set $S_j\in \mathcal{S}$ we create four agents $s_j^1,s_j^2,s_j^3, t_j$ and an object~$p_j$. For each element $i\in \mathcal{X}$ we create an object $o_i$ and two dummy agents $d_i^1,d_i^2$. Altogether we have $12\hat{n}+6\hat{n}=18\hat{n}$ agents and $3\hat{n}+3\hat{n}=6\hat{n}$ objects. The preferences are described and illustrated in Figure~\ref{fig:basecons}.

\begin{figure}
\begin{minipage}[l]{0.4\textwidth}
\begin{tabular}{rll}
    $s_j^{\ell}:$  &  $o_{j_{\ell}}\succ p_j$ & for $j\in [3\hat{n}]$, $\ell \in [3]$ \\
    $t_j:$ & $p_j$ & for $j\in [3\hat{n}]$\\
    $d_i^1,d_i^2:$ & $o_i$ & for $i\in [3\hat{n}]$
\end{tabular}
\end{minipage}\hspace{3mm}
\begin{minipage}[r]{0.56\textwidth}
    \centering
    \includegraphics[height=0.26\textheight]{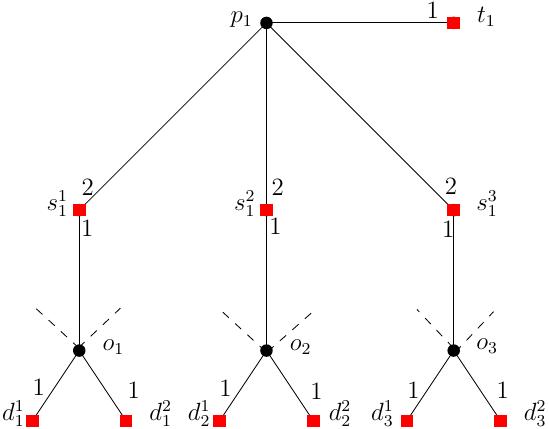}
    \end{minipage}
    \caption{The preferences and a  graph illustration of a gadget of a set $S_1=\{ 1,2,3\}$ in the proof of Theorem~\ref{th:khai}. Red squares are agents, black disks are objects, and the numbers on the edges indicate the preferences to the left of the graph. Dashed edges run to other gadgets. Each agent finds at most two objects acceptable.}
    \label{fig:basecons}
\end{figure}

\noindent We prove that there is a $(5\hat{n}+1)$-stable matching $M$ in $I'$ if and only if there is an exact 3-cover in~$I$.

\begin{claim}
If $I$ admits an exact 3-cover, then $I'$ admits a $(5\hat{n}+1)$-stable matching.
\end{claim}
\begin{claimproof}
Suppose that $S_{l^1},\dots,S_{l^{\hat{n}}}$ form an exact 3-cover. Construct a matching $M$ as follows: for each $j\in [3\hat{n}]$ match $t_j$ with $p_j$. For each $j\in \{ l^1,\dots,l^{\hat{n}}\}$ match $s_j^{\ell}$ with $o_{j_{\ell}}$ for $\ell \in [3]$. As each object is covered exactly once, $M$ is a matching.

We claim that $M$ is $(5\hat{n}+1)$-stable. Due to Observation~\ref{lemma:ks-charac} it is enough to show that at most $5\hat{n}$ objects are envied by any agent, so the improvement graph has at most $5\hat{n}$ objects with non-zero degree. For each $j\in \{ l^1,\dots,l^{\hat{n}}\}$, the object $p_j$ is not envied by anyone, as all of $s_j^1,s_j^2,s_j^3$ got their best object and $t_j$ got $p_j$. Hence, at most $2\hat{n}$ objects of type $p_j$ and at most $3\hat{n}$ objects of type $o_i$ are envied, proving our claim. 
\end{claimproof}

\begin{claim}
If $I'$ admits a $(5\hat{n}+1)$-stable matching, then $I$ admits an exact 3-cover.
\end{claim}
\begin{claimproof}
Let $M$ be a $(5\hat{n}+1)$-stable matching. First we prove that in $M$ at most $5\hat{n}$ objects are envied by any agent, because we can construct a matching $M'$ that gives all envied objects to an agent who envies them. Each object $o_i$ is envied in any matching by at least one of $d_i^1$ and $d_i^2$ and can be given to the envious agent in~$M'$. Regarding envied objects of type $p_j$, one agent only finds at most one $p_j$ object acceptable, which implies that each envied $p_j$ can be assigned to an envious agent in~$M'$. Therefore, $M'$ is indeed a matching.

As at most $5\hat{n}$ objects can be envied and all $o_i$ objects are envied in $M$, at most $2\hat{n}$ envied objects are of type~$p_j$. This implies that at least $\hat{n}$ objects of type~$p_j$ are not envied by any agent. As these objects are the first choices of their $s_j^{\ell}$, $\ell \in [3]$ agents, those agents must all get their first-choice object of type~$o_i$. As $M$ is a matching, these sets constitute an exact 3-cover.
\end{claimproof}
\end{proof}

\begin{theorem}
\label{th:majhai}
\maj\ is NP-complete even if each agent finds at most two objects acceptable.
\end{theorem}
\begin{proof}
We extend our hardness reduction in the proof of Theorem~\ref{th:khai}. To show the hardness of \maj, we add $8\hat{n}$ more gadgets to the instance $I'$, each consisting of $3$ agents and $3$ objects, such that all 3 agents have the same preference order over their three corresponding objects. This gadget is a small version of our example instance in Example~\ref{ex:cond}. It is easy to see that in any such gadget, there is a 3-stable matching, but there is no 2-stable matching. Hence, the new instance has $18\hat{n}+8\cdot 3\hat{n}=42\hat{n}$ agents and there is a $5\hat{n}+1+8\cdot 2\hat{n} = (21\hat{n}+1)$-stable matching if and only if there is an exact 3-cover. 
\end{proof}

We now apply a more general scaling argument than in the proof of Theorem~\ref{th:majhai} to show that finding a $cn$-stable matching is NP-complete for any non-trivial choice of~$c$.

\begin{theorem}
\label{th:khai_allk}
    \khai\ is NP-complete for any fixed constant $0<c<1$.
\end{theorem}
\begin{proof}
Our construction in the proof of Theorem~\ref{th:khai} can be extended to \khai\ by scaling the instance. This scaling happens through the addition of instances that either admit no $cn$-stable matching even for a high $c$, or admit a $cn$-stable matching even for a low~$c$. The instance in Example~\ref{ex:cond} admits no $(n-1)$-stable matching. Constructing an instance with a 1-stable matching is easy: one $(1,o_1)$ edge suffices. For any constant $c$, we can add sufficiently many of these scaling instances to construct an instance that admits a $cn$-stable matching if and only if the original instance admits a majority stable matching.
\end{proof}

\subsubsection{Complete preferences}\hfill\\

We now extend our hardness proof in Theorem~\ref{th:khai_allk} to cover the case of complete preference lists as well.

\begin{theorem}
\label{th:kha}
    \kha\ is NP-complete for any fixed constant $0<c<1$. In particular, \cmaj\ is NP-complete.
\end{theorem}
\begin{proof}

We start with an instance $I$ of \khai\ and create an instance $I'$ of \kha\ as follows. For each agent $i\in N$, we add a dedicated dummy object~$d_i$. Then, we extend the preferences of the agents in a standard manner: we append their dedicated dummy object to the end of their original preference list, followed by all other objects in an arbitrary order. 

Suppose first that there is a $cn$-stable matching $M$ in~$I$. We construct a matching $M'$ in $I'$ by keeping all edges of $M$ and assigning each unmatched agent in $M$ to their dummy object in~$M'$. As the improvement graph of $M'$ in $I'$ is the same as the improvement graph of $M$ in $I$, $M'$ is $cn$-stable in~$I'$.

Now suppose that there is a  $cn$-stable matching $M'$ in~$I'$.
Construct a matching $M$ by taking $M=M'\cap E$, where $E$ denotes the edges of the acceptability graph of $I$. We claim that $M$ is $cn$-stable in~$I$. Suppose that matching $M''$ is preferred to $M$ by more than $k$ agents in~$I$. We claim that $M''$ is preferred to $M'$ as well by more than $k$ agents in~$I'$. Indeed, each agent who prefers $M''$ to $M$ must be assigned to an object in $M''$ that is acceptable to them in $I$, so it is better than their dummy object. As each agent either obtains the same object in $M$ and $M'$, or is unassigned in $M$ and matched to their dummy object in $M'$, we get that each improving agent prefers $M''$ to~$M'$.
\end{proof}

\subsubsection{Ties in the preference lists}\hfill\\

From Theorems~\ref{th:khai_allk} and~\ref{th:kha} follows that \khati\ and \khat\ are both NP-complete for any constant $0<c<1$. Therefore, we investigate the preference restrictions \textsc{dc} and \textsc{sti}.

\begin{theorem}
\label{th:ksmallties}
\khati\ is NP-complete even if each agent puts their two acceptable objects into a single tie.
This holds for any fixed constant $0<c\le \frac{1}{2}$.
\end{theorem}
\begin{proof}
We use the same construction as in Theorem~\ref{th:khai}, on $18\hat{n}$ agents, except that the agents with two objects in their preference lists are indifferent between these objects. Let $I$ be an instance of \textsc{x3c} and $I'$ be the constructed instance. We will prove that there is a $(5\hat{n}+1)$-stable matching $M$ in $I'$ if and only if there is an exact 3-cover in~$I$. As the constructions are very similar, only small modifications to the proof of Theorem~\ref{th:khai} are needed.

\begin{claim}
If $I$ admits an exact 3-cover, then $I'$ admits a $(5\hat{n}+1)$-stable matching.
\end{claim}
\begin{claimproof}
Suppose that $S_{l^1},\dots,S_{l^{\hat{n}}}$ form an exact 3-cover. Construct a matching $M$ as follows: for each $j\in [3\hat{n}]$, match $t_j$ with $p_j$. For each $j\in \{ l^1,\dots,l^{\hat{n}}\}$ match $s_j^{\ell}$ with $o_{j_{\ell}}$ for $\ell \in [3]$. 

To prove that $M$ is $(5\hat{n}+1)$-stable, it is enough to show that at most $5\hat{n}$ objects are envied, so the improvement graph has at most $5\hat{n}$ objects with non-zero degree. For each $j\in \{ l^1,\dots,l^{\hat{n}}\}$, the object $p_j$ is not envied by any agent: all of $s_j^1,s_j^2,s_j^3$, and $t_j$ got an object. Hence, there are at most $2\hat{n}$ objects of type $p_j$ and at most $3\hat{n}$ objects of type $o_i$ envied, proving our claim.
\end{claimproof}

\begin{claim}
If $I'$ admits a $(5\hat{n}+1)$-stable matching, then $I$ admits an exact 3-cover.
\end{claim}
\begin{claimproof} Let $M$ be a $(5\hat{n}+1)$-stable matching. First we prove that in $M$ at most $5\hat{n}$ objects are envied by any agent, because we can construct a matching $M'$ that gives all envied objects to an agent who envies them. 
Each object $o_i$ is envied in any matching by at least one of $d_i^1$ and $d_i^2$ and can be given to the envious agent in~$M'$. Regarding envied objects of type $p_j$, one agent only finds at most one $p_j$ object acceptable, which implies that each envied $p_j$ can be assigned to an envious agent in~$M'$. Therefore, $M'$ is indeed a matching.


As at most $5\hat{n}$ objects can be envied and all $o_i$ objects are envied in $M$, at most $2\hat{n}$ envied objects are of type~$p_j$. This implies that at least $\hat{n}$ objects of type~$p_j$ are not envied by any agent. As these objects are acceptable to their $s_j^{\ell}$, $\ell \in [3]$ agents, those agents must all get their other acceptable object of type~$o_i$. As $M$ is a matching, these sets constitute an exact 3-cover.
\end{claimproof}

Finally, we can add sufficiently many instances to $I'$, each of which consists of two agents and one object, with both agents only considering their one object acceptable. In this small instance, there is no $cn$-stable matching with $c\le \frac{1}{2}$, so hardness holds for any such $c>\frac{5}{18}$. Similarly, we can add instances with one agent and one object to prove the hardness for $0 < c < \frac{5}{18}$.
\end{proof}

We complement Theorem~\ref{th:ksmallties} with positive results in some restricted cases. For example, while \khati\ is hard if each agent is indifferent between their acceptable objects (restriction \textsc{sti}), \tmaj\ becomes solvable in this case.

\begin{theorem}
\label{thm:alltie+dichHA}
The following statements hold:
\begin{enumerate}
    \item If each agent's preference list is a single tie in \tmaj, a majority stable matching exists and can be found in $O(\sqrt{n}m)$ time.
    \item If each agent's preference list is dichotomous and complete in \ctmaj\ with $|O|\ge |N|$, then a majority stable matching exists and can be found in $O(\sqrt{n}m)$ time.
    \item  If each agent's preference list is dichotomous and complete in \ctmaj\ with $|O|<|N|$, a $\left(\frac{2n}{3}+1\right)$-stable matching exists and can be found in $O(\sqrt{n}m)$ time.
\end{enumerate} 
\end{theorem}

\begin{proof}
We prove each statement separately.
\begin{enumerate}
    \item We claim that any maximum size matching $M$ is majority stable in this case. 
If $|M|\ge \frac{n}{2}$, then $M$ is obviously majority stable, as at least half of the agents obtain an object. If $|M|<\frac{n}{2}$, then $M$ is also majority stable. Suppose there is a matching $M'$, where more than $\frac{n}{2}$ agents improve. Then $|M'| > \frac{n}{2}$, contradicting the fact that $M$ with $|M|<\frac{n}{2}$ is a maximum size matching.

\item We first find a maximum-size matching in the graph containing only the first-choice edges for each agent. Then we extend this matching by assigning every so far unmatched agent a second-choice object. Note that this is possible as the preferences are dichotomous and complete. In this matching, agents can only improve by switching to a first-choice object from a second-choice object. This cannot be the case for a majority of the agents, because we started with a maximum-size matching in the graph containing only the first-choice edges for each agent. So the constructed matching is majority stable.

\item If $|O| \leq \frac{2}{3}n$, then the statement is obvious. Otherwise, let $M$ be a maximum matching (so each object is matched as the preferences are complete) that matches as many agents to a first-choice object as possible---constructed the same way as in the previous case. If $M$ matches at least $\frac{n}{3}$ agents to first-choice objects, then at most $\frac{2}{3}n$ agents can improve, so $M$ is $\left(\frac{2}{3}n+1\right)$-stable. Otherwise, let $x <\frac{n}{3}$ denote the number of agents who get a first-choice object in $M$ and $y = |O|-x$ denote the number of agents who get a second-choice object. Clearly, at most $x+(n-x-y)=n-y$ agents can improve in a different matching, because at most $x$ can improve by getting a first choice (by the choice of $M$) and at most $n-x-y$ can improve by a second choice (only the ones unmatched). Hence, if $M$ is not $(\frac{2}{3}n+1)$-stable, then $y<\frac{n}{3}$ and $x<\frac{n}{3}$, contradicting that we assumed that the number of objects $x+y$ is at least $\frac{2}{3}n$.\qedhere
\end{enumerate}
\end{proof}

The first two statements and the fact that $k$-stable matchings are also $k+1$-stable imply the following.  

\begin{mycorollary}
\khati\ is solvable in $O(\sqrt{n}m)$ time for any $c> \frac{1}{2}$ when each preference list is a single tie.
The same holds for \khat\ with dichotomous preferences.
\end{mycorollary}

Our last result for the house allocation model is valid for the preference restriction  \textsc{dc}.

\begin{theorem}
\label{th:khat}
\khat\ with dichotomous preferences is NP-complete \begin{enumerate}
    \item for any fixed $c\le \frac{1}{2}$ even if $|N| = |O|$;
    \item if $|O|<|N|$, then even for any fixed constant $c\le \frac{2}{3}$.
\end{enumerate}

\end{theorem}
\begin{proof}
We extend the reduction used in the proof of Theorem~\ref{th:ksmallties}. First we prove hardness for $k = \frac{5}{18}$, then we pad the instance to cover $0 < c \leq \frac{1}{2}$, and finally, we extend our proof to $c\le \frac{2}{3}$ with the assumption that $|O|<|N|$.

We add another $12\hat{n}$ dummy objects to the \khati\ instance and extend the preferences of the agents such that they rank their originally acceptable objects first, and all the other objects---including the $12\hat{n}$ dummy objects---second. Since the acceptability graph is complete, and $|N| = |O|$, it is clear that if there exists a $(5\hat{n}+1)$-stable matching, then there exists a $(5\hat{n}+1)$-stable matching in which all agents are matched.

We claim that there is a $(5\hat{n}+1)$-stable matching in this modified instance $I'$ if an only if there is a $(5\hat{n}+1)$-stable matching in the original \khati\ instance~$I$. 

Indeed, if there is such a matching $M$ in $I$, then all original objects are matched and extending $M$ in $I'$ by matching the unmatched agents to the dummy objects arbitrarily produces a matching $M'$ that must be $(5\hat{n}+1)$-stable, as only those agents can improve who got a dummy object, and only by getting a first-choice, hence originally acceptable, object. 

If there is a $(5\hat{n}+1)$-stable agent-complete matching in $I'$, then among the agents who are matched to second-choice objects, at most $5\hat{n}$ can improve simultaneously. Specifically, there must be a matching in the graph $G'$ induced by the best object edges for each agent, such that among the unmatched agents at most $5\hat{n}$ can be matched in $G'$. Such a matching must be $(5\hat{n}+1)$-stable in~$I$.

To prove hardness for all $0 < c \le \frac{1}{2}$, we either add agents with a dedicated first-ranked object for small $c$ or multiple instances with two agents and two objects, one as their only first-ranked object and the other one ranked second by both of them for larger $c$. At least half of those agents will always be able to improve even after we make the preferences complete by adding all remaining edges and setting the preferences for the originally unacceptable objects as second.

Finally, suppose that $|O|<|N|$. We extend the hardness for $\frac{1}{2}\le c\le \frac{2}{3}$. For this, we further add sufficiently many copies of an instance with two objects and three agents, such that one object is best, while the other object is second-best for all three agents. Finally, we add the remaining edges as second-ranked edges. Let $x$ be the number of copies we added and let $I''$ be the new instance. We claim that there is a $\left(5\hat{n}+\frac{2}{3}x+1\right)$-stable matching in $I''$ if and only if there is a $\left(5\hat{n}+1\right)$-stable matching in~$I'$. In one direction, if $M'$ is $\left(5\hat{n}+1\right)$-stable in $I'$, then we create $M''$ by matching two out of the three agents in each added instance to their corresponding two objects. Then, $M''$ is $\left(5\hat{n}+\frac{2}{3}x+1\right)$-stable, because at most $\frac{1}{3}x+5\hat{n}$ agents can improve by getting first-ranked objects and at most $\frac{1}{3}x$ agents can improve by getting matched. In the other direction, if there is a $\left(5\hat{n}+\frac{2}{3}x+1\right)$-stable matching $M''$ in~$I''$, then we claim that $M''$ restricted to $I'$ is $\left(5\hat{n}+1\right)$-stable in~$I'$. Otherwise, more than $5\hat{n}+1$ agents could improve by getting a first-ranked object among the ones in $I'$, and at least $\frac{x}{3}$ agents could improve by getting a first-ranked object among the newly added agents and there are at least $\frac{x}{3}$ unmatched agents, who all could improve with any object left, which is altogether more than $5\hat{n}+\frac{2}{3}x+1$, contradiction.
\end{proof}

\section{The marriage and roommates models}
\label{se:mr}

In this section, we settle most complexity questions in the marriage and roommates models. Just as in Section~\ref{se:ha}, we first prove that verification can be done in polynomial time in Section~\ref{se:smsrver} and then turn to the existence problems in Section~\ref{se:smsrex}.

\subsection{Verification}
\label{se:smsrver}

\begin{theorem}
    Verifying whether a matching is $k$-stable can be done in $O(n^3)$ time for any $k \in \mathbb{N}$, both in the marriage and roommates models, even the preference lists contain ties.
\label{th:versmsr}
\end{theorem}

\begin{proof}
Let $M$ be a matching in a \ksri\ instance. We create an edge weight function $\omega$, where $\omega (e)$ is the number of end vertices of $e$ who prefer $e$ to~$M$. From the definition of $k$-stability follows that $M$ is $k$-stable if and only if maximum weight matchings in this graph have weight less than~$k$. Such a matching can be computed in $O(n^3)$ time \citep{Gab76}.
\end{proof}

\subsection{Existence}
\label{se:smsrex}

To 
show that a $\left( \frac{5}{6}n+1 \right)$-stable matching exists even in the roommates model with ties, we first introduce the concept of stable partitions, which generalizes the notion of a stable matching. Let $(N, \succ )$ be a stable roommates instance. A \emph{stable partition} of $(N, \succ )$ is a permutation $\pi :N\to N$ such that for each $i\in N$: \begin{enumerate}
    \item if $\pi (i)\ne \pi^{-1}(i)$, then $(i,\pi (i)),(i,\pi^{-1}(i))\in E$ and $\pi (i) \succ_{i}\pi^{-1}(i)$;
    \item for each $(i,j) \in E$, if $\pi (i)=i$ or $j \succ_{i}\pi^{-1}(i)$, then $\pi^{-1}(j) \succ_{j}i$.
\end{enumerate}

A stable partition defines a set of edges and cycles. \citet{tan1991necessary} showed that a stable partition always exists and that one can be found in polynomial time. He also showed that a stable matching exists if and only if a stable partition does not contain any odd cycle.

\begin{theorem}
\label{th:ksrpol}
    For any $c>\frac{5}{6}$, a $cn$-stable matching exists in \ksri\ and can be found efficiently.
\end{theorem}
\begin{proof}
We present an algorithm to construct such a matching. We apply Tan's algorithm to obtain a stable partition. Then, for each odd cycle of length at least three, we remove an arbitrary vertex from the cycle. This leaves us with components that are either even cycles, or paths on an even number of  vertices, or single vertices, as the odd cycles become paths on an even number of  vertices after the removal of a vertex. In all these components except the single vertices, we choose a perfect matching. 

Denote the matching obtained in this way by~$M$. Clearly, $M$ has at most $\frac{n}{2} $ edges. We claim that $M$ is $\left( \frac{5}{6}n+1 \right)$-stable. Let $M'$ be any matching. Observe that if an edge $e\in M'$ has the property that both of its end vertices prefer it to $M$, then $e$ must be adjacent to one of the deleted vertices. This is because otherwise both end vertices of $e$ would prefer $e$ to their worst partner (if there is any) in the stable partition, but this is a contradiction to it being a stable partition in the first place. Also, observe that the number of deleted vertices $x$ just the number of odd cycles, which is at most $\frac{n}{3}$. Combining these, we get that the number of agents who prefer $M'$ to $M$ is at most $2x +(\frac{n}{2}-x)=\frac{n}{2} +x\le \frac{5}{6}n$.

Hence, $M$ is $\left(\frac{5}{6}n+1\right)$-stable, which proves our statement.
\end{proof}
As \ksr, \ksmi, and \ksm\ are subcases of \ksri\, the existence of a $cn$-stable matching follows from Theorem~\ref{th:ksrpol}. In the marriage model, even majority stable matchings are guaranteed to exist.

\begin{theorem}
    In the marriage model, a majority stable matching exists and can be found in $O(m)$ time. Thus, for any $c> \frac{1}{2}$, a $cn$-stable matching exists and can be found in $O(m)$ time.
    \label{th:smpol}
\end{theorem}
\begin{proof}
    Even in the presence of ties, a stable matching can be found in $O(m)$ time with the Gale-Shapley aglorithm (by breaking the ties arbitrarily). Such a matching is $cn$-stable for any $c>\frac{1}{2}$. To see this, let $M$ be stable and suppose that there is a matching $M'$ where more than half of the agents improve. Then, there must be two agents, who both improve and are matched to each other to $M'$, as all improving agents must get a partner. But they would be a blocking pair to $M$, contradicting the stability of $M$.
\end{proof}

As majority stable matchings always exist in the marriage model, it is natural to ask whether we can find a majority stable matching that is a maximum size matching as well. We denote the problem of deciding if such a matching exists by \maxmajsmi. In the case of complete preferences, a maximum size and majority stable matching exists and can be found efficiently, as popular matchings are both majority stable and maximum size. Otherwise, the situation is less preferable, as the following theorem shows.

\begin{theorem}
\label{th:npsmsr}
The following problems are NP-complete:
\begin{enumerate}
    \item \maxmajsmi\ even if $|U| = |W|$;
    \item  \ksm\ and \ksmi\ for any fixed constant $c \leq \frac{1}{2}$ even if $|U| = |W|$;
    \item  \ksr\ and \ksri\ for any fixed constant $c \leq \frac{2}{3}$.
\end{enumerate}
\end{theorem}

\begin{proof}
By Theorem~\ref{th:versmsr}, all decision problems on the existence of a $cn$-stable matching in the marriage and roommates models are in NP. We again reduce from \textsc{x3c}. Let $I$ be an instance of \textsc{x3c} and let $\mathcal{S}=\{ S_1,\dots,S_{3\hat{n}}\}$. First, we create an instance $I'$ of \textsc{sm}, which admits a $(16\hat{n}+1)$-stable matching if and only if it admits a maximum size $(16\hat{n}+1)$-stable matching if and only if $I$ admits an exact 3-cover. This instance will be the basis of all three reductions.

Let us denote the two classes of the agents we create in $I'$ by $U$ and $W$.
For each set $S_j\in S$ we create five agents $s_j^1,s_j^2,s_j^3, y_j,q_j$ in $U$, and five agents $c_j^1,c_j^2,c_j^3, x_j,p_j$ in~$W$. For each element $i\in \mathcal{X}$ we create five agents $b_i, d_i,e_i,f_i',g_i'$ in $U$ and five agents $a_i, d_i',e_i',f_i,g_i$ in~$W$. Altogether we have $n=60\hat{n}$ agents.

The preferences are described in Figure~\ref{fig:ksmi}. 
Let $S_j=\{ j_1,j_2,j_3\}$ and for an element $i$, let $S_{i_1},S_{i_2},S_{i_3}$ be the three sets that contain $i$ with $\ell_1,\ell_2,\ell_3$ denoting the position of $i$ in the sets $S_{i_1},S_{i_2},S_{i_3}$, respectively. 

\begin{figure}
    \centering
    
    \begin{tabular}{rlrl}
    $s_j^{\ell}:$ &  $a_{j_{\ell}}\succ p_j \succ c_j^{\ell}$ & $c_j^{\ell}:$ & $b_{j_{\ell}}\succ q_j\succ s_j^{\ell}$  \\
     $q_j:$ & $x_j\succ c_j^1\succ c_j^2\succ c_j^3$ & $p_j:$ & $y_j\succ s_j^1\succ s_j^2\succ s_j^3$ \\
    $y_j:$ & $p_j$  & $x_j:$ & $q_j$ \\
    $b_i:$ & $f_i\succ c_{i_1}^{\ell_1}\succ c_{i_2}^{\ell_2}\succ c_{i_3}^{\ell_3}\succ g_i$ \qquad \qquad \qquad & $a_i:$ & $d_i \succ s_{i_1}^{\ell_1}\succ s_{i_2}^{\ell_2}\succ s_{i_3}^{\ell_3}\succ e_i$ \\
    $d_i:$ & $a_i\succ d_i'$ & $f_i:$ & $b_i\succ f_i'$ \\
    $e_i:$ & $a_i\succ e_i'$ & $g_i:$ & $b_i\succ g_i'$ \\
    $f_i':$ & $f_i$ & $d_i':$ & $d_i$ \\
    $g_i':$ & $g_i$ & $e_i':$ & $e_i$ \\
\end{tabular}\vspace{3mm}
    \includegraphics[height=0.75\textheight]{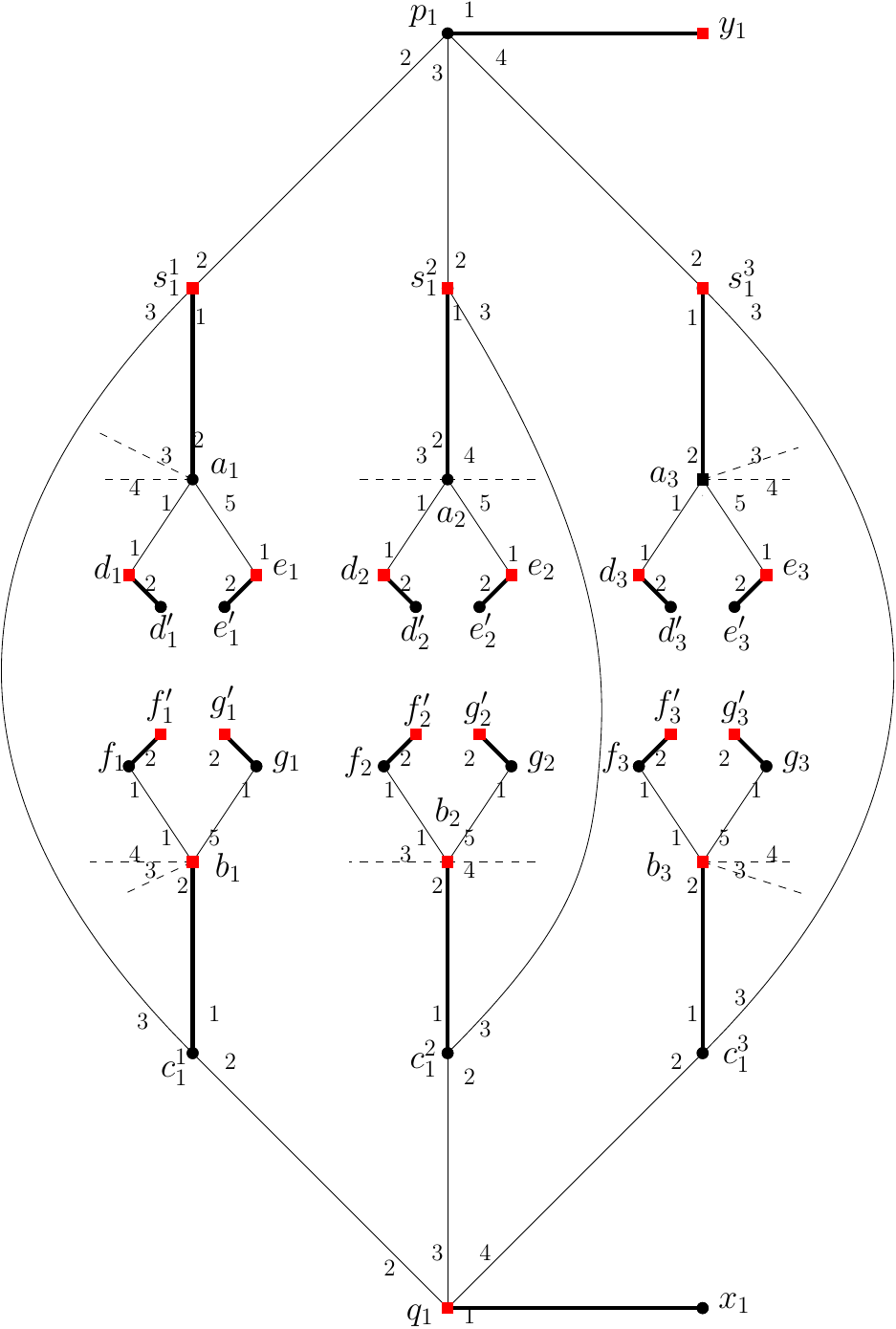}
    \caption{The construction for Theorem~\ref{th:npsmsr}. The preference list of each created agent, $i,j\in [3\hat{n}]$, $\ell \in [3]$, can be seen above the graph. The figure illustrates the gadget of a set $S_1=\{ 1,2,3\}$. The numbers on the edges indicate the preferences. Thick edges denote the matching edges if $S_1$ is in the exact-3-cover, dashed edges run to other gadgets.}
    \label{fig:ksmi}
\end{figure}

\begin{claim}
If $I$ admits an exact 3-cover, then $I'$ admits a maximum size $(16\hat{n}+1)$-stable matching.
\end{claim}
\begin{claimproof}
To the exact 3-cover $S_{l^1},\dots, S_{l^{\hat{n}}}$ in $I$ we create a matching $M$ in $I'$ as follows---see Figure~\ref{fig:ksmi} as well.
\begin{itemize}
    \item For each $i,j\in [3\hat{n}]$, we add the edges $(y_j,p_j),(d_i,d_i'),(e_i,e_i'),(q_j,x_j),(f_i',f_i),(g_i',g_i)$. 
    \item For each $j\in \{ l^1,\dots,l^{\hat{n}}\}$ and $\ell \in [3]$, we add the edges $(s_j^{\ell},a_{j_{\ell}}),(b_{j_{\ell}},c_j^{\ell})$.
    \item For each $j\notin \{ l^1,\dots,l^{\hat{n}}\}$ and $\ell \in [3]$, we add the edge $(s_j^{\ell},c_j^{\ell})$. 
\end{itemize}
As the sets formed an exact cover, $M$ is a matching, and because all agents are matched, $M$ is a maximum size matching in~$I'$. We next show that in any matching $M'$, at most $16\hat{n}$ agents can improve.
\begin{itemize}
    \item Agents of type $s_j^{\ell},c_j^{\ell}$ can improve with an agent of type $p_j$ or $q_j$, respectively, if $(s_j^{\ell},c_j^{\ell}) \in M$. As $S_{l^1},\dots, S_{l^{\hat{n}}}$ was an exact cover, at most $2\hat{n}$ of the $s_j^{\ell}$ agents can improve with an agent of type $p_j$---only those with $j\notin \{ l^1,\dots ,l^{\hat{n}}\}$---and similarly, at most $2\hat{n}$ of the $c_j^{\ell}$ agents can improve with an agent of type~$q_j$.
    \item Agents of type $a_i$ or $b_i$ can improve.
    \item All other agents can only improve by switching to an agent of type $a_i$ or~$b_i$. Even if all agents of type  $a_i,b_i$ improve and they get a partner who also improves with them, $2 \cdot 2 \cdot 3 \hat{n} = 12\hat{n}$ agents can improve as or through agents of type $a_i$ or~$b_i$.
\end{itemize}
Therefore, altogether at most $4\hat{n} + 12\hat{n}=16\hat{n}$ agents can improve, concluding the proof of our claim.
\end{claimproof}

\begin{claim}
If $I'$ admits a maximum size $(16\hat{n}+1)$-stable matching, then $I'$ admits a $(16\hat{n}+1)$-stable matching.
\end{claim}

\begin{claim}
If $I'$ admits a $(16\hat{n}+1)$-stable matching, then $I$ admits an exact 3-cover.
\end{claim}
\begin{claimproof}
Suppose that $I'$ admits a $(16\hat{n}+1)$-stable matching $M$, but there is no exact 3-cover in~$I$. We count the number of agents who can improve.
\begin{itemize}
    \item For any $i$, if $a_i$ is not matched to $d_i$ in $M$, then both $a_i$ and $d_i$ can improve if they get matched together, and otherwise both $d_i'$ and $e_i$ can improve with the edges $(d_i,d_i'),(e_i,a_i)$. Therefore, in any $M$, at least $6\hat{n}$ agents from the set $\{ d_i',d_i,e_i,a_i \mid i\in [3\hat{n}] \}$ can improve among themselves and similarly, at least $6\hat{n}$ agents from $\{ f_i',f_i,g_i,b_i \mid i\in [3\hat{n}] \}$ can improve among themselves.
    \item As no exact 3-cover exists in~$I$, there are at least $2\hat{n}+1$ indices $j\in [3\hat{n}]$, such that at least one of $\{ s_j^1,s_j^2,s_j^3\}$ is not matched to an agent of type $a_i$ in $M$ 
    and similarly, there are at least $2\hat{n}+1$ indices $j'\in [3\hat{n}]$, such that at least one of $\{ c_{j'}^1,c_{j'}^2,c_{j'}^3\}$ is not matched to an agent of type $b_i$ in~$M$. For each such $j$ and $j'$, at least one of the agents in the set $\{ s_j^1,s_j^2,s_j^3,y_j,p_j\}$ and at least one of the agents in $\{ c_{j'}^1,c_{j'}^2,c_{j'}^3,x_{j'},q_{j'}\}$ can improve among themselves: if $(y_j,p_j)\notin M$ or $(q_{j'},x_{j'})\notin M$, then $y_j$ and $p_j$ or $x_{j'}$ and $q_{j'}$ can both improve with each other, otherwise there is an $s_j^{\ell}$ agent with a $c_j^{\ell}$ agent or a $c_{j'}^{\ell}$ agent with an $s_{j'}^{\ell}$ agent, who could improve with $p_j$ or $q_{j'}$, respectively.
\end{itemize}
 Altogether at least $12\hat{n}+4\hat{n}+2=16\hat{n}+2$ agents can improve, a contradiction to the $(16\hat{n}+1)$-stability of~$M$. 
\end{claimproof}

Now we use this construction to prove the hardness of all three problems. 

\begin{enumerate}
    \item For \maxmajsmi, we add a path $P$ with $28\hat{n}+2$ vertices to $I'$, to have $88\hat{n}+2$ agents. Path $P$ has an even number of vertices and therefore a unique maximum size matching~$M_P$. However, we create the preferences of the agents such that they all prefer their other edge in $P$, except for the end vertices. Hence, by switching to the edges in $P$ outside of $M_P$, $28\hat{n}$ agents improve. So in this instance, there is a maximum size majority stable matching, and therefore a $(44\hat{n}+2)$-stable matching if and only if $I'$ admits a matching, where at most $16\hat{n}+1$ agents can improve, which happens if and only if $I$ admits an exact 3-cover, as we have seen in the main part of the proof.

    \item We distinguish two cases for \ksm. First let $\frac{4}{15}\le c \leq  \frac{1}{2}$. In this case, we add to $I'$ some paths on four vertices, such that the vertices with degree two prefer the middle edge. Hence, for any matching $M$, in any of these paths, at least two agents can improve. By adding enough of these paths, we can get an instance that has a $cn$-stable matching if and only if there is a $(16\hat{n}+1)$-stable matching in~$I'$.
    
    For $c<\frac{4}{15}$, we add a sufficiently large instance that admits a 1-stable matching---a union of edges suffices.

    \item We apply a similar case distinction for the last statement on \ksr. First let $\frac{4}{15}\le c \leq  \frac{2}{3}$. Now we add to $I'$ some (but an even number of) triangles with cyclic preferences. For any matching $M$, in any such triangle at least two out of the three agents can improve. This still holds if we add edges between the vertices of the first two, the vertices of the second two, \ldots, the vertices of the last two triangles, which are worse for both sides. This leads to an instance in which there is a complete $cn$-stable matching if there is any $cn$-stable matching. Hence, by adding sufficiently many of these triangles, we can get an instance that has a $cn$-stable matching if and only if there is a $(16\hat{n}+1)$-stable matching in~$I'$.
    
    For $c<\frac{4}{15}$, we add a sufficiently large instance that admits a 1-stable matching.
\end{enumerate}
The second and the third reductions remain intact for complete preferences, when we add the remaining agents to the end of the preference lists. It is easy to see that if there is a $cn$-stable matching $M'$ with the extended preferences, then there is one where each agent obtains an original partner---we just project $M'$ to the original acceptability graph to get~$M$. If there is a matching $M''$, where at least $k$ agents improve from $M$, then these $k$ agents must also prefer $M''$ to $M'$, which contradicts the $cn$-stability of~$M'$.

In the other direction, we have by our reduction that if there is a $cn$-stable matching $M$ with incomplete preferences, then there is one that is complete as well. Hence, agents can only improve with original edges, so if a matching $M''$ would be better for at least $cn$ agents than $M$ with complete preferences, then so would it be with incomplete preferences as well, which is a contradiction.
\end{proof}

We now turn to the preference-restricted cases \textsc{dc} and \textsc{sti}.
\begin{theorem}
\label{thm:alltiesrpos}
    If each preference list consists of a single tie, then a $cn$-stable matching can be found $O(\sqrt{n}m)$ time for any $c>\frac{1}{3}$.
\end{theorem}
\begin{proof}
Let $M$ be a maximum size matching in the acceptability graph $G=(N,E)$, which can be found in $O(\sqrt{n}m)$ time \citep{MV80}. We claim that $M$ is $\left(\frac{n}{3}+1\right)$-stable. Suppose there is a matching $M'$, where at least $\frac{n}{3}+1$ agents improve. As the preferences consist of a single tie, this is only possible if all of them were unmatched in $M$ and they can be matched simultaneously. By the famous Gallai-Edmonds decomposition, we know that there is a set $X\subset N$ such that:
\begin{itemize}
    \item each maximum size matching matches every vertex of $X$;
    \item every vertex in $X$ is matched to a vertex in distinct odd components in $G\setminus X$;
    \item $G$ has exactly $q(X)-|X|$ unmatched vertices, each of which are in distinct odd components in $G\setminus X$, where $q(X)$ denotes the number of odd components in $G\setminus X$.
\end{itemize}

As $\frac{n}{3}+1$ agents can improve, we get that $q(X)-|X|>\frac{n}{3}$. Let $s$ denote the number of singleton components in $G\setminus X$. If such a vertex was unmatched in $M$, then it can only improve by getting matched to someone in~$X$. Let $x\le s$ be the number of agents in singleton components who improve in~$M'$. Then, $|X|\ge x$ by our observation. Also, at least $\frac{n}{3}-x+1$ agents must improve from the other odd components, which all have size at least 3. Furthermore, there must be at least $x$ odd components, whose vertices are all matched by $M$ according to the Gallai-Edmonds characterization. Hence, we get that the number of vertices of $G$ is at least $x +3 \cdot (\frac{n}{3}-x+1)+2\cdot x\ge n+3$, which is a contradiction.
\end{proof}

\begin{theorem}
\label{thm:dichcompsrpos}
    If the preferences are complete and dichotomous, then a $cn$-stable matching exists and can be found in $O(n^3)$ time for any $c>\frac{1}{2}$.
\end{theorem}
\begin{proof}
  We define the edge weight function $w(e)$ as the number of end vertices of $e$ who rank $e$ best. For a matching $M$, $w(M)$ is equal to the number of agents who get a first choice. Let $M$ be a maximum weight matching with respect to~$w(e)$. Then we extend $M$ to a maximum size matching, which matches all agents as the preferences are complete. 
Suppose there is a matching $M'$, where more than $\frac{n}{2}$ agents improve. As agents can only improve by getting a first-choice partner, we get that $w(M')>\frac{n}{2}>w(M)$, contradiction. 
\end{proof}

\begin{theorem}
\label{thm:alltiesrneg}
 $c$-\textsc{smi} is NP-complete for any fixed constant $c\le \frac{1}{3}$ even if each preference list is a single tie.
\end{theorem}
\begin{proof}
We start with the reduction in the proof of Theorem~\ref{th:ksmallties}. There, one side of the graph consisted of objects, which now correspond to agents. To ensure that each agent's list is a single tie, all acceptable agents are tied in the preference lists of these new agents.

In the proof of Theorem~\ref{th:ksmallties} we showed that in the base instance with $18\hat{n}$ agents and $3\hat{n}$ objects it is NP-complete to decide if there exists a $\left(5\hat{n}+1\right)$-stable matching. 
As the reduction had the property that any inclusion-wise maximal matching assigned all objects, it follows that in our new instance with two-sided preferences there is a $\left(5\hat{n}+1\right)$-stable matching if and only if there was one with one-sided preferences, as we can always suppose that the agents on the smaller side cannot improve.

To extend hardness to $c\le \frac{5}{21}$, we just add a sufficiently large instance that admits a 1-stable matching. To extend hardness to any $\frac{5}{21}<c\le \frac{1}{3}$, we add paths with 3 vertices to the instance. As in each such path, there is always an unmatched agent, at least one third of these agents can improve, hence by adding sufficiently many copies, the theorem follows.
\end{proof}

\begin{theorem}
\label{thm:dichcompsrneg}
    $c$-\textsc{smc} is NP-complete for any fixed constant $c\le \frac{1}{2}$, even if the preferences are dichotomous.
\end{theorem}
\begin{proof}
    We extend the instance in Theorem~\ref{thm:alltiesrneg}---which had $3\hat{n}$ agents on one side and $18\hat{n}$ on the other side---by completing the acceptability graph via adding the remaining edges such that they are ranked second for each end vertex. It is clear that deciding if there is a $\left(5\hat{n}+1\right)$-stable matching remains NP-complete.

    To show hardness for smaller $c$ values, we pad the instance by adding pairs of agents who rank each other first, and all other agents second. For $\frac{5}{21}<c\le \frac{1}{2}$, we first add paths with 4 vertices, such that the middle two vertices only rank each other first, while the end vertices rank their only neighbor first. In this small instance, at least half of the agents can always improve and this fact remains true even after making the acceptability graph complete by adding the remaining edges as second best for all agents. 
\end{proof}

\section{Conclusion and open questions}
\label{se:con}

We have settled the main complexity questions on the verification and existence of $k$-stable and majority stable matchings in all three major matching models. We derived that the existence of a $k$-stable solution is the easiest to guarantee in the marriage model, while it cannot be guaranteed for any non-trivial $k$ at all in the house allocation model. Only one case remains partially open: in the roommates model, the existence of a $cn$-stable solution is guaranteed for $c \ge \frac{5}{6}$ (Theorem~\ref{th:ksrpol}), whereas NP-completeness was only proved for $c<\frac{2}{3}$ (Theorem~\ref{th:npsmsr}, point~3). We conjecture polynomial solvability for $\frac{2}{3}\le  c < \frac{5}{6}$.

A straightforward direction of further research would be to study the strategic behavior of the agents. It is easy to prove that $k$-stability, as stability and popularity, is fundamentally incompatible with strategyproofness. However, mechanisms that guarantee strategyproofness for a subset of agents might be developed. Another rather game-theoretic direction would be to investigate the price of $k$-stability.

A much more applied line of research involves computing the smallest $k$ for which implemented solutions of real-life matching problems are $k$-stable. For example: given a college admission pool and its stable outcome, how many of the students could have gotten into a better college in another matching? We conjecture that the implemented solution can only be improved for a little fraction of the applicants simultaneously. Simulations supporting this could potentially strengthen the trust in the system.

\subsubsection*{Acknowledgments.} We thank Barton E.\ Lee for drawing our attention to the concept of majority stability and for producing thought-provoking example instances. We also thank the reviewers of the paper for their suggestions that helped improving the presentation of the paper. 
Haris Aziz acknowledges the support from  NSF-CSIRO grant on ‘Fair Sequential Collective Decision-Making’. 
Gergely Csáji acknowledges the financial support by the Hungarian Academy of Sciences, Momentum Grant No. LP2021-1/2021, and by the Hungarian Scientific Research Fund, OTKA, Grant No.\ K143858. \'{A}gnes Cseh's work was supported by OTKA grant K128611 and the J\'anos Bolyai Research Fellowship.

\bibliographystyle{abbrvnat}
\bibliography{mybib}

\begin{thebibliography}{54}
\providecommand{\natexlab}[1]{#1}
\providecommand{\url}[1]{\texttt{#1}}
\expandafter\ifx\csname urlstyle\endcsname\relax
  \providecommand{\doi}[1]{doi: #1}\else
  \providecommand{\doi}{doi: \begingroup \urlstyle{rm}\Url}\fi

\bibitem[Abraham et~al.(2004)Abraham, Cechl{\'a}rov{\'a}, Manlove, and
  Mehlhorn]{ACM+04}
D.~J. Abraham, K.~Cechl{\'a}rov{\'a}, D.~F. Manlove, and K.~Mehlhorn.
\newblock {P}areto optimality in house allocation problems.
\newblock In \emph{Proceedings of the 15th International Symposium on
  Algorithms and Computation, {ISAAC 2004}}, pages 3--15. Springer, 2004.

\bibitem[Abraham et~al.(2007)Abraham, Irving, Kavitha, and Mehlhorn]{AIK+07}
D.~J. Abraham, R.~W. Irving, T.~Kavitha, and K.~Mehlhorn.
\newblock Popular matchings.
\newblock \emph{SIAM Journal on Computing}, 37:\penalty0 1030--1045, 2007.

\bibitem[Aziz et~al.(2013)Aziz, Brandt, and Harrenstein]{ABH13}
H.~Aziz, F.~Brandt, and P.~Harrenstein.
\newblock {P}areto optimality in coalition formation.
\newblock \emph{Games and Economic Behavior}, 82:\penalty0 562--581, 2013.

\bibitem[Aziz et~al.(2018)Aziz, Chen, Gaspers, and Sun]{ACGS18}
H.~Aziz, J.~Chen, S.~Gaspers, and Z.~Sun.
\newblock Stability and {P}areto optimality in refugee allocation matchings.
\newblock In \emph{AAMAS'18}, pages 964--972, 2018.

\bibitem[Balinski and S{\"o}nmez(1999)]{BS99}
M.~Balinski and T.~S{\"o}nmez.
\newblock A tale of two mechanisms: Student placement.
\newblock \emph{Journal of Economic Theory}, 84\penalty0 (1):\penalty0 73--94,
  1999.

\bibitem[Balliu et~al.(2022)Balliu, Flammini, Melideo, and Olivetti]{BFO22}
A.~Balliu, M.~Flammini, G.~Melideo, and D.~Olivetti.
\newblock On {P}areto optimality in social distance games.
\newblock \emph{Artificial Intelligence}, 312:\penalty0 103768, 2022.

\bibitem[Bhattacharya et~al.(2015)Bhattacharya, Hoefer, Huang, Kavitha, and
  Wagner]{BHH+15}
S.~Bhattacharya, M.~Hoefer, C.-C. Huang, T.~Kavitha, and L.~Wagner.
\newblock Maintaining near-popular matchings.
\newblock In \emph{International Colloquium on Automata, Languages, and
  Programming}, pages 504--515. Springer, 2015.

\bibitem[Bir{\'o} and Gudmundsson(2020)]{BG20}
P.~Bir{\'o} and J.~Gudmundsson.
\newblock Complexity of finding {P}areto-efficient allocations of highest
  welfare.
\newblock \emph{European Journal of Operational Research}, 2020.

\bibitem[Bir\'o et~al.(2010)Bir\'o, Irving, and Manlove]{BIM10}
P.~Bir\'o, R.~W. Irving, and D.~F. Manlove.
\newblock Popular matchings in the marriage and roommates problems.
\newblock In \emph{CIAC '10: Proceedings of the 7th International Conference on
  Algorithms and Complexity}, volume 6078 of \emph{Lecture Notes in Computer
  Science}, pages 97--108. Springer, 2010.

\bibitem[Bogomolnaia and Moulin(2004)]{BM04}
A.~Bogomolnaia and H.~Moulin.
\newblock Random matching under dichotomous preferences.
\newblock \emph{Econometrica}, 72\penalty0 (1):\penalty0 257--279, 2004.

\bibitem[Braun et~al.(2010)Braun, Dwenger, and K\"ubler]{BDK10}
S.~Braun, N.~Dwenger, and D.~K\"ubler.
\newblock Telling the truth may not pay off: an empirical study of centralized
  university admissions in {G}ermany.
\newblock \emph{The B.E.\ Journal of Economic Analysis and Policy}, 10, article
  22, 2010.

\bibitem[Bullinger(2020)]{B20}
M.~Bullinger.
\newblock {P}areto-optimality in cardinal hedonic games.
\newblock In \emph{AAMAS'20}, pages 213--221, 2020.

\bibitem[Cechl{\'a}rov{\'a} et~al.(2014)Cechl{\'a}rov{\'a}, Eirinakis, Fleiner,
  Magos, Mourtos, and Potpinkov{\'a}]{CEFM+14}
K.~Cechl{\'a}rov{\'a}, P.~Eirinakis, T.~Fleiner, D.~Magos, I.~Mourtos, and
  E.~Potpinkov{\'a}.
\newblock {P}areto optimality in many-to-many matching problems.
\newblock \emph{Discrete Optimization}, 14:\penalty0 160--169, 2014.

\bibitem[Cechl{\'a}rov{\'a} et~al.(2016)Cechl{\'a}rov{\'a}, Eirinakis, Fleiner,
  Magos, Manlove, Mourtos, Ocel{\'a}kov{\'a}, and Rastegari]{CEFM+16}
K.~Cechl{\'a}rov{\'a}, P.~Eirinakis, T.~Fleiner, D.~Magos, D.~Manlove,
  I.~Mourtos, E.~Ocel{\'a}kov{\'a}, and B.~Rastegari.
\newblock {P}areto optimal matchings in many-to-many markets with ties.
\newblock \emph{Theory of Computing Systems}, 59\penalty0 (4):\penalty0
  700--721, 2016.

\bibitem[Chen and S\"onmez(2002)]{CS02}
Y.~Chen and T.~S\"onmez.
\newblock Improving efficiency of on-campus housing: an experimental study.
\newblock \emph{American Economic Review}, 92:\penalty0 1669--1686, 2002.

\bibitem[{Condorcet}(1785)]{Con85}
M.~{Condorcet}.
\newblock \emph{Essai sur l'application de l'analyse \`a la probabilit\'e des
  d\'ecisions rendues \`a la pluralit\'e des voix}.
\newblock L'Imprimerie Royale, 1785.

\bibitem[Cseh(2017)]{Cse17}
{\'A}.~Cseh.
\newblock Popular matchings.
\newblock \emph{Trends in Computational Social Choice}, 105\penalty0 (3), 2017.

\bibitem[Cseh and Kavitha(2021)]{CK21}
{\'A}.~Cseh and T.~Kavitha.
\newblock Popular matchings in complete graphs.
\newblock \emph{Algorithmica}, 83\penalty0 (5):\penalty0 1493--1523, 2021.

\bibitem[Cseh et~al.(2017)Cseh, Huang, and Kavitha]{CHK17}
{\'A}.~Cseh, C.-C. Huang, and T.~Kavitha.
\newblock Popular matchings with two-sided preferences and one-sided ties.
\newblock \emph{SIAM Journal on Discrete Mathematics}, 31\penalty0
  (4):\penalty0 2348--2377, 2017.

\bibitem[Darmann(2013)]{Dar13}
A.~Darmann.
\newblock Popular spanning trees.
\newblock \emph{International Journal of Foundations of Computer Science},
  24\penalty0 (05):\penalty0 655--677, 2013.

\bibitem[Darmann(2018)]{Dar18}
A.~Darmann.
\newblock A social choice approach to ordinal group activity selection.
\newblock \emph{Mathematical Social Sciences}, 93:\penalty0 57--66, 2018.

\bibitem[Ehrgott and Nickel(2002)]{EN02}
M.~Ehrgott and S.~Nickel.
\newblock On the number of criteria needed to decide {P}areto optimality.
\newblock \emph{Mathematical Methods of Operations Research}, 55\penalty0
  (3):\penalty0 329--345, 2002.

\bibitem[Elkind et~al.(2020)Elkind, Fanelli, and Flammini]{EFF20}
E.~Elkind, A.~Fanelli, and M.~Flammini.
\newblock Price of {P}areto optimality in hedonic games.
\newblock \emph{Artificial Intelligence}, 288:\penalty0 103357, 2020.

\bibitem[Faenza et~al.(2019)Faenza, Kavitha, Powers, and Zhang]{FKP+19}
Y.~Faenza, T.~Kavitha, V.~Powers, and X.~Zhang.
\newblock Popular matchings and limits to tractability.
\newblock In \emph{Proceedings of SODA '19: the Thirtieth Annual ACM-SIAM
  Symposium on Discrete Algorithms}, pages 2790--2809. ACM-SIAM, 2019.

\bibitem[Florenzano et~al.(2006)Florenzano, Gourdel, and Jofr{\'e}]{FGJ06}
M.~Florenzano, P.~Gourdel, and A.~Jofr{\'e}.
\newblock Supporting weakly {P}areto optimal allocations in infinite
  dimensional nonconvex economies.
\newblock \emph{Economic Theory}, 29:\penalty0 549--564, 2006.

\bibitem[Gabow(1976)]{Gab76}
H.~Gabow.
\newblock An efficient implementations of {E}dmonds' algorithm for maximum
  matching on graphs.
\newblock \emph{Journal of the ACM}, 23\penalty0 (2):\penalty0 221--234, 1976.

\bibitem[Gai et~al.(2007)Gai, Lebedev, Mathieu, de~Montgolfier, Reynier, and
  Viennot]{GLM+07}
A.-T. Gai, D.~Lebedev, F.~Mathieu, F.~de~Montgolfier, J.~Reynier, and
  L.~Viennot.
\newblock Acyclic preference systems in {P2P} networks.
\newblock In A.~Kermarrec, L.~Boug{\'{e}}, and T.~Priol, editors,
  \emph{Euro-Par '07: Proceedings of the 13th International Euro-Par Conference
  (European Conference on Parallel and Distributed Computing)}, volume 4641 of
  \emph{Lecture Notes in Computer Science}, pages 825--834. Springer, 2007.

\bibitem[Gale and Shapley(1962)]{GS62}
D.~Gale and L.~S. Shapley.
\newblock College admissions and the stability of marriage.
\newblock \emph{American Mathematical Monthly}, 69:\penalty0 9--15, 1962.

\bibitem[G{\"a}rdenfors(1975)]{Gar75}
P.~G{\"a}rdenfors.
\newblock Match making: assignments based on bilateral preferences.
\newblock \emph{Behavioural Science}, 20:\penalty0 166--173, 1975.

\bibitem[Garey and Johnson(1979)]{GJ79}
M.~R. Garey and D.~S. Johnson.
\newblock \emph{Computers and Intractability}.
\newblock Freeman, San Francisco, CA., 1979.

\bibitem[Gupta et~al.(2021)Gupta, Misra, Saurabh, and Zehavi]{GMS+21}
S.~Gupta, P.~Misra, S.~Saurabh, and M.~Zehavi.
\newblock Popular matching in roommates setting is {NP}-hard.
\newblock \emph{ACM Transactions on Computation Theory}, 13\penalty0 (2), 2021.
\newblock ISSN 1942-3454.

\bibitem[Hopcroft and Karp(1973)]{HK73}
J.~Hopcroft and R.~Karp.
\newblock A $n^{5/2}$ algorithm for maximum matchings in bipartite graphs.
\newblock \emph{SIAM Journal on Computing}, 2:\penalty0 225--231, 1973.

\bibitem[Huang and Kavitha(2013{\natexlab{a}})]{HK13}
C.-C. Huang and T.~Kavitha.
\newblock Popular matchings in the stable marriage problem.
\newblock \emph{Information and Computation}, 222:\penalty0 180--194,
  2013{\natexlab{a}}.

\bibitem[Huang and Kavitha(2013{\natexlab{b}})]{HK13a}
C.-C. Huang and T.~Kavitha.
\newblock Near-popular matchings in the roommates problem.
\newblock \emph{SIAM Journal on Discrete Mathematics}, 27\penalty0
  (1):\penalty0 43--62, 2013{\natexlab{b}}.

\bibitem[Huang and Kavitha(2021)]{HK21}
C.-C. Huang and T.~Kavitha.
\newblock Popularity, mixed matchings, and self-duality.
\newblock \emph{Mathematics of Operations Research}, 46\penalty0 (2):\penalty0
  405--427, 2021.

\bibitem[Huang et~al.(2011)Huang, Kavitha, Michail, and Nasre]{HKMN11}
C.-C. Huang, T.~Kavitha, D.~Michail, and M.~Nasre.
\newblock Bounded unpopularity matchings.
\newblock \emph{Algorithmica}, 61\penalty0 (3):\penalty0 738--757, 2011.

\bibitem[Irving(1985)]{Irv85}
R.~W. Irving.
\newblock An efficient algorithm for the ``stable roommates'' problem.
\newblock \emph{Journal of Algorithms}, 6:\penalty0 577--595, 1985.

\bibitem[Kavitha(2014)]{Kav14}
T.~Kavitha.
\newblock A size-popularity tradeoff in the stable marriage problem.
\newblock \emph{SIAM Journal on Computing}, 43:\penalty0 52--71, 2014.

\bibitem[Kavitha et~al.(2011)Kavitha, Mestre, and Nasre]{KMN11}
T.~Kavitha, J.~Mestre, and M.~Nasre.
\newblock Popular mixed matchings.
\newblock \emph{Theoretical Computer Science}, 412:\penalty0 2679--2690, 2011.

\bibitem[Kavitha et~al.(2022{\natexlab{a}})Kavitha, Kir{\'a}ly, Matuschke,
  Schlotter, and Schmidt-Kraepelin]{KKM+22a}
T.~Kavitha, T.~Kir{\'a}ly, J.~Matuschke, I.~Schlotter, and
  U.~Schmidt-Kraepelin.
\newblock Popular branchings and their dual certificates.
\newblock \emph{Mathematical Programming}, 192\penalty0 (1):\penalty0 567--595,
  2022{\natexlab{a}}.

\bibitem[Kavitha et~al.(2022{\natexlab{b}})Kavitha, Kir{\'a}ly, Matuschke,
  Schlotter, and Schmidt-Kraepelin]{KKM+22b}
T.~Kavitha, T.~Kir{\'a}ly, J.~Matuschke, I.~Schlotter, and
  U.~Schmidt-Kraepelin.
\newblock The popular assignment problem: when cardinality is more important
  than popularity.
\newblock In \emph{SODA '22: Proceedings of the 2022 Annual ACM-SIAM Symposium
  on Discrete Algorithms}, pages 103--123. SIAM, 2022{\natexlab{b}}.

\bibitem[Kraiczy et~al.(2021)Kraiczy, Cseh, and Manlove]{KCM21}
S.~Kraiczy, {\'A}.~Cseh, and D.~Manlove.
\newblock On weakly and strongly popular rankings.
\newblock In \emph{Proceedings of the 20th International Conference on
  Autonomous Agents and MultiAgent Systems}, AAMAS '21, page 1563–1565,
  Richland, SC, 2021. International Foundation for Autonomous Agents and
  Multiagent Systems.
\newblock ISBN 9781450383073.

\bibitem[Manlove(2013)]{Man13}
D.~F. Manlove.
\newblock \emph{Algorithmics of Matching Under Preferences}.
\newblock World Scientific, 2013.

\bibitem[McCutchen(2008)]{McC08}
R.~M. McCutchen.
\newblock The least-unpopularity-factor and least-unpopularity-margin criteria
  for matching problems with one-sided preferences.
\newblock In E.~Laber, C.~Bornstein, L.~Nogueira, and L.~Faria, editors,
  \emph{Proceedings of LATIN '08: the 8th Latin-American Theoretical
  Informatics Symposium}, volume 4957 of \emph{Lecture Notes in Computer
  Science}, pages 593--604. Springer Berlin Heidelberg, 2008.

\bibitem[Micali and Vazirani(1980)]{MV80}
S.~Micali and V.~Vazirani.
\newblock An {$O(\sqrt{|V|} \cdot |E|)$} algorithm for finding maximum matching
  in general graphs.
\newblock In \emph{Proceedings of FOCS '80: the 21st Annual IEEE Symposium on
  Foundations of Computer Science}, pages 17--27. IEEE Computer Society, 1980.

\bibitem[Perach et~al.(2008)Perach, Polak, and Rothblum]{PPR08}
N.~Perach, J.~Polak, and U.~G. Rothblum.
\newblock A stable matching model with an entrance criterion applied to the
  assignment of students to dormitories at the {T}echnion.
\newblock \emph{International Journal of Game Theory}, 36:\penalty0 519--535,
  2008.

\bibitem[Peters(2016)]{Pet16}
D.~Peters.
\newblock Complexity of hedonic games with dichotomous preferences.
\newblock In \emph{Proceedings of the AAAI Conference on Artificial
  Intelligence}, volume 30:1, 2016.

\bibitem[Roth(1982)]{Rot82}
A.~E. Roth.
\newblock Incentive compatibility in a market with indivisible goods.
\newblock \emph{Economics Letters}, 9:\penalty0 127--132, 1982.

\bibitem[Roth and Sotomayor(1990)]{RS90}
A.~E. Roth and M.~A.~O. Sotomayor.
\newblock \emph{Two-Sided Matching: A Study in Game-Theoretic Modeling and
  Analysis}, volume~18 of \emph{Econometric Society Monographs}.
\newblock Cambridge University Press, 1990.

\bibitem[Ruangwises and Itoh(2021)]{RI21}
S.~Ruangwises and T.~Itoh.
\newblock Unpopularity factor in the marriage and roommates problems.
\newblock \emph{Theory of Computing Systems}, 65\penalty0 (3):\penalty0
  579--592, 2021.

\bibitem[Tan(1991)]{tan1991necessary}
J.~J. Tan.
\newblock A necessary and sufficient condition for the existence of a complete
  stable matching.
\newblock \emph{Journal of Algorithms}, 12\penalty0 (1):\penalty0 154--178,
  1991.

\bibitem[Thakur(2021)]{Tha21}
A.~Thakur.
\newblock Combining social choice and matching theory to understand
  institutional stability.
\newblock In \emph{The 25th Annual ISNIE / SIOE Conference}. Society for
  Institutional and Organizational Economics, 2021.

\bibitem[{van Zuylen} et~al.(2014){van Zuylen}, Schalekamp, and
  Williamson]{VSW14}
A.~{van Zuylen}, F.~Schalekamp, and D.~Williamson.
\newblock Popular ranking.
\newblock \emph{Discrete Applied Mathematics}, 165:\penalty0 312--316, 2014.

\bibitem[Warburton(1983)]{War83}
A.~R. Warburton.
\newblock Quasiconcave vector maximization: connectedness of the sets of
  {P}areto-optimal and weak {P}areto-optimal alternatives.
\newblock \emph{Journal of Optimization Theory and Applications}, 40:\penalty0
  537--557, 1983.

\end{thebibliography}
\end{document}